\documentclass[12pt,english]{article}
\pdfoutput=1
\usepackage{lmodern}
\usepackage{bbm}
\DeclareMathAlphabet{\mathbbold}{U}{bbold}{m}{n} 

\usepackage[T1]{fontenc}
\usepackage[latin9]{inputenc}
\usepackage{geometry}
\usepackage{color}
\usepackage{cite}
\usepackage[english]{babel}
\usepackage{amsmath}
\usepackage{amsthm}
\usepackage{amssymb}
\usepackage{cancel}
\usepackage{graphicx}
\usepackage{esint}
\usepackage{float}
\usepackage[unicode=true,pdfusetitle,
 bookmarks=true,bookmarksnumbered=false,bookmarksopen=false,
 breaklinks=false,pdfborder={0 0 1},backref=false,colorlinks=true]
 {hyperref}
\hypersetup{
 citecolor=blue,linkcolor=black,urlcolor=blue}

\makeatletter

 \@ifundefined{textcolor}{}
 {%
   \definecolor{BLACK}{gray}{0}
   \definecolor{WHITE}{gray}{1}
   \definecolor{RED}{rgb}{1,0,0}
   \definecolor{GREEN}{rgb}{0,1,0}
   \definecolor{BLUE}{rgb}{0,0,1}
   \definecolor{CYAN}{cmyk}{1,0,0,0}
   \definecolor{MAGENTA}{cmyk}{0,1,0,0}
   \definecolor{YELLOW}{cmyk}{0,0,1,0}
 }

\usepackage{babel}

\makeatletter
\def\simgt{\mathrel{\lower2.5pt\vbox{\lineskip=0pt\baselineskip=0pt
           \hbox{$>$}\hbox{$\sim$}}}}
\def\simlt{\mathrel{\lower2.5pt\vbox{\lineskip=0pt\baselineskip=0pt
           \hbox{$<$}\hbox{$\sim$}}}}
\makeatother

\newcommand{\be}{\begin{equation}}
\newcommand{\bea}{\begin{eqnarray}}
\newcommand{\ee}{\end{equation}}
\newcommand{\eea}{\end{eqnarray}}
\newcommand\beq{\begin{equation}}
\newcommand\eeq{\end{equation}}

\newcommand{\Tr}{{\rm Tr\,}}

\newcommand{\Eq}[1]{Eq.~(\ref{#1})}
\newcommand{\Eqs}[2]{Eqs.~(\ref{#1}) and (\ref{#2})}
\newcommand{\Sec}[1]{Sec.~\ref{#1}}

\newcommand{\Fig}[1]{Fig.~\ref{#1}}
\newcommand{\Ref}[1]{Ref.~\cite{#1}}

\newcommand{\ket}[1]{|#1\rangle}
\newcommand{\bra}[1]{\langle #1|}

\newtheorem{theorem}{Theorem}
\newtheorem{lemma}{Proposition}

\begin{document}
\interfootnotelinepenalty=10000
\baselineskip=18pt
\hfill CALT-TH-2016-008
\hfill

\vspace{1.5cm}
\thispagestyle{empty}
\begin{center}
{\LARGE\bf
Entanglement Conservation, ER=EPR, and a\\ New Classical Area Theorem for Wormholes
}

\bigskip\vspace{1.5cm}{
{\large
Grant N. Remmen,$^1$ Ning Bao,$^{1,2}$ and Jason Pollack$^{1}$}
} \\[7mm]
{\it
$^1$Walter Burke Institute for Theoretical Physics\\ and \\
$^2$Institute for Quantum Information and Matter,\\[-1mm] California Institute of Technology, Pasadena, CA 91125, USA}

\let\thefootnote\relax\footnote{e-mail: 
\url{gremmen@theory.caltech.edu},
\url{ningbao@theory.caltech.edu},
\url{jpollack@theory.caltech.edu}} \\
 \end{center}
\bigskip
\centerline{\large\bf Abstract}

\begin{quote}\small
We consider the question of entanglement conservation in the context of the ER=EPR correspondence equating quantum entanglement with wormholes. In quantum mechanics, the entanglement between a system and its complement is conserved under unitary operations that act independently on each; ER=EPR suggests that an analogous statement should hold for wormholes. We accordingly prove a new area theorem in general relativity: for a collection of dynamical wormholes and black holes in a spacetime satisfying the null curvature condition, the maximin area for a subset of the horizons (giving the largest area attained by the minimal cross section of the multi-wormhole throat separating the subset from its complement) is invariant under classical time evolution along the outermost apparent horizons. The evolution can be completely general, including horizon mergers and the addition of classical matter satisfying the null energy condition. This theorem is the gravitational dual of entanglement conservation and thus constitutes an explicit characterization of the ER=EPR duality in the classical limit.
\end{quote}

\setcounter{footnote}{0}

\newpage
\tableofcontents

\section{Introduction}\label{sec:Introduction}

All of the states of a quantum mechanical theory are on the same footing when considered as vectors in a Hilbert space: any state can be transformed into any other state by the application of a unitary operator. When the Hilbert space can be decomposed into subsystems, however, there is a natural way to categorize them: by the entanglement entropy of the reduced density matrix of a subsystem constructed from the states. Entanglement between two subsystems is responsible for the ``spooky action at a distance'' often considered a characteristic feature of quantum mechanics: measuring some property of a subsystem determines the outcome of measuring the same property on another entangled subsystem, even a causally disconnected one. 

It is well known that this seeming nonlocality does not lead to violations of causality. It cannot be used to send faster-than-light messages \cite{Dieks} and in fact it is impossible for any measurement to determine whether the state is entangled (see, e.g., \Ref{nielsen2010quantum}). Similarly, it is impossible to alter the entanglement between a system and its environment (that is, to change the entanglement entropy of the reduced density matrix of the system) by acting purely on the degrees of freedom in the system or by adding more unentangled degrees of freedom. A number of well-established properties, such as monogamy \cite{Coffman:1999jd} and strong subadditivity \cite{Lieb:1973cp}, constrain the entanglement entropy of subsystems created from arbitrary factorizations of the Hilbert space. 

Although entanglement entropy is a fundamental quantity, it is typically very difficult to compute in field theories, where working directly with the reduced density matrix can be computationally intractable, although important progress has been made in certain conformal field theories \cite{Holzhey:1994we,Calabrese:2004eu} and more generally along lightsheets for interacting quantum field theories \cite{Bousso:2014uxa}. The AdS/CFT correspondence \cite{AdSCFT,Witten,MAGOO}, however, allows us to transform many field-theoretic questions to a gravitational footing. In particular, the Ryu-Takayanagi formula \cite{RT} equates the entanglement entropy of a region for a state in a conformal field theory living on the boundary of an asymptotically AdS spacetime to the area of a minimal surface with the same boundary as that region in the spacetime corresponding to that CFT state. Using this identification of entropy with area, a number of ``holographic entanglement inequalities'' have been proven \cite{Hayden:2011ag,Bao:2015bfa}, some reproducing and some stronger than the purely quantum mechanical entanglement inequalities.

Motivated in part by AdS/CFT, as well as a number of older ideas in black hole thermodynamics \cite{BHLaws,Bekenstein} and holography \cite{Holography1,Holography2,Holography3}, Maldacena and Susskind have recently conjectured \cite{ER=EPR} an ER=EPR correspondence, an exact duality between entangled states (Einstein-Podolsky-Rosen \cite{EPR} pairs) and so-called ``quantum wormholes'', which reduce in the classical general relativistic limit to two-sided black holes (Einstein-Rosen \cite{ER} bridges, i.e.,  wormholes). In a series of recent papers, we have considered the implications of this correspondence in the purely classical regime. In this limit, if the ER=EPR duality holds true, certain statements in quantum mechanics about entangled states should match directly with statements in general relativity about black holes and wormholes \cite{Susskind}, with the same assumptions required on both sides. We indeed previously found two beautiful and nontrivial detailed correspondences: the no-cloning theorem in quantum mechanics corresponds to the no-go theorem for topology change in general relativity \cite{Bao:2015nqa} and the unobservability of entanglement corresponds to the undetectability of the presence or absence of a wormhole \cite{Bao:2015nca}. 

In this paper, we extend this correspondence to a direct equality between the entanglement entropy and a certain invariant area, which we define, of a geometry containing classical black holes and wormholes. We follow a long tradition of clarifying general relativistic dynamics using area theorems \cite{Hawking:1971vc,Hawking&Ellis,Krolak,Bousso:2014sda,Bousso:2015mqa}, which hold that various areas of interest satisfy certain properties under time evolution. Our strategy is to show that the area in question remains unchanged under dynamics constituting the gravitational analogue of applying tensor product operators to an individual system and its complement. We show that, just as entanglement entropy cannot be changed by acting on the subsystem and its complement separately, this area is not altered by merging pairs of black holes or wormholes or by adding classical (unentangled) matter. The area we consider is chosen to be that of a maximin surface \cite{Hubeny:2007xt,Wall:2012uf} for a collection of wormhole horizons, a time-dependent generalization of the Ryu-Takayanagi minimal area, which again establishes that the entanglement entropy is also conserved under these operations. At least for asymptotically AdS spacetimes, our result constitutes an explicit characterization of the ER=EPR correspondence in the classical limit. Moreover, our theorem is additionally interesting from the gravitational perspective alone, as it constitutes a new area law within general relativity.

This paper is structured as follows. In Section \ref{sec:entanglement}, we review the simple quantum mechanical fact that entanglement is conserved under local operations. In Section \ref{sec:Maximin}, we define the maximin surface and review its properties. In Section \ref{sec:gravity}, we prove our desired general relativistic theorem. Finally, we discuss the implications of our result and conclude in Section \ref{sec:Conclusions}.

\section{Conservation of Entanglement}\label{sec:entanglement}

Consider a Hilbert space ${\cal H}$ that can be written as a tensor product of two factors ${\cal H}_{\rm L}$ and ${\cal H}_{\rm R}$ to which we will refer as ``right'' and ``left'', though they need not have any spatial interpretation. For a state $\ket{\psi}\in{\cal H}$, let us define the reduced density matrix associated with ${\cal H}_{\rm L}$ as $\rho_{\rm L} = \Tr_{{\cal H}_{\rm R}} \ket{\psi}\bra{\psi}$ and use this to define the entanglement entropy between the right and left sides of the Hilbert space:
\be 
S(L) = S(R) = -\Tr_{{\cal H}_{\rm L}} \rho_{\rm L} \log \rho_{\rm L}.
\ee
It is straightforward to see that adding more unentangled degrees of freedom to ${\cal H}_{\rm L}$ will not affect the entanglement entropy, as by construction this does not introduce new correlations between ${\cal H}_{\rm L}$ and ${\cal H}_{\rm R}$. This is particularly clear to see by using the equivalence of $S(L)$ and $S(R)$ for pure states, as adding in further unentangled degrees of freedom will maintain the purity of the joint system.

Now let us consider the effect on $S(L)$ of applying a unitary $U = U_{\rm L} \otimes U_{\rm R}$ to $\ket{\psi}$. As $\Tr_{{\cal H}_{\rm R}} U=U_{\rm L}$, we can consider only the action of $U_{\rm L}$ on $\rho_{\rm L}$, as $U_{\rm R}$ acts trivially in ${\cal H}_{\rm L}$. This transforms $S(L)$ into
\be
S(L)= -\Tr_{{\cal H}_{\rm L}} U_{\rm L}\rho_{\rm L}U^{\dagger}_{\rm L} \log \left(U_{\rm L}\rho_{\rm L}U^{\dagger}_{\rm L}\right).
\ee
 One can at this point expand the logarithm by power series, with individual terms of the form
\be
S_n(L)=-\Tr_{{\cal H}_{\rm L}}c_n U_{\rm L}\rho_{\rm L}U^{\dagger}_{\rm L} \left( \mathbbm{1}-U_{\rm L}\rho_{\rm L}U^{\dagger}_{\rm L}\right)^n
\ee
for some real $c_n$. For each term in the expansion of the product, all but the first $U_{\rm L}$ and the last $U^{\dagger}_{\rm L}$ will cancel as $U^{\dagger}_{\rm L}U_{\rm L}=\mathbbm{1}$. Finally, by cyclicity of the trace, the remaining $U_{\rm L}$ and $U^{\dagger}_{\rm L}$ will also cancel, leaving $S_n(L)$ invariant. Thus, $S(L)$ remains invariant under unitary transformations of the form $U = U_{\rm L} \otimes U_{\rm R}$. This is the statement of conservation of entanglement.

\section{The Maximin Surface}\label{sec:Maximin}
A holographic characterization of the entanglement entropy begins with its calculation on a constant-time slice, where the Ryu-Takayanagi (RT) formula \cite{RT} holds:
\be
S(H)=\frac {A_H}{4G\hbar}.
\ee
This relates the area $A_H$ of the minimal surface subtending a region $H$ to the entanglement entropy of that region with its complement. When the region is a complete boundary, this reduces to the minimal surface homologous to the region. For example, in a hypothetical static wormhole geometry, the entanglement entropy between the two ends would be given by the minimal cross-sectional area of the wormhole.

This method of computing entanglement entropy on a constant-time slice for static geometries was generalized by the Hubeny-Rangamani-Takayanagi (HRT) proposal \cite{Hubeny:2007xt}. The key insight here was that in general there do not exist surfaces that have minimal area in time, as small perturbations can decrease the area. The new proposal was that the area now scales as the smallest extremal area surface, as opposed to the minimal area. The homology condition mentioned previously remains in this prescription.

The maximin proposal \cite{Wall:2012uf} gives an explicit algorithm for the implementation of the HRT prescription.  In the following definitions, we will closely follow the conventions used by Wall \cite{Wall:2012uf}. We define $C[H,\Gamma]$ to be the codimension-two surface of minimal area homologous to $H$ anchored to $\partial H$ that lies on any complete achronal (i.e., spacelike or null) slice $\Gamma$. Note that $C[H,\Gamma]$ can refer to any minimal area surface that exists on $\Gamma$. Next, the maximin surface $C[H]$ is defined as any of the $C[H,\Gamma]$ with the largest area when optimized over all achronal surfaces $\Gamma$. When multiple such candidate maximin surfaces exist, we refine the definition of $C[H]$ to mean any such surface that is a local maximum as a functional over achronal surfaces $\Gamma.$ In the HRT proposal, the entanglement of $H$ with its complement in the boundary is given by $S(H)=\mathrm{area}[C[H]]/4G\hbar$.

As an example, for a wormhole geometry in which we are computing the entanglement entropy between the two horizons of the ER bridge, $\partial H$ is trivial and the homology condition means that $C[H,\Gamma]$ is the surface of minimal cross-sectional area on an achronal surface $\Gamma$ in the interior causal diamond of the horizons. Then the maximin surface $C[H]$ is a $C[H,\Gamma]$ with $\Gamma$ chosen such that the area is maximized.

Such surfaces can be shown to exist for large classes of spacetimes and in particular $C[H]$ can be proven to be equal to the extremal HRT surface for spacetimes obeying the null curvature condition, which is given by
\be
R_{\mu\nu}k^\mu k^\nu \geq 0,\label{eq:NCC}
\ee
 where $k^\mu$ is any null vector and $R_{\mu\nu}$ is the Ricci tensor.\footnote{For spacetimes satisfying the Einstein equation $R_{\mu\nu} - R \,g_{\mu\nu}/2 = 8\pi G \, T_{\mu\nu}$ for energy-momentum tensor $T_{\mu\nu}$, the null curvature condition is equivalent to the null energy condition $T_{\mu\nu} k^\mu k^\nu \geq 0$.} As HRT is a covariant method of calculating entanglement entropy, the maximin construction is therefore manifestly covariant as well.

Maximin surfaces in general have some further nice properties, proven in \Ref{Wall:2012uf}: they have smaller area than the causal surface (the edge of the causal domain of dependence associated with bulk causality), they move monotonically outward as the boundary region increases in size, they obey strong subadditivity, and they also obey monogamy of mutual information, but not necessarily other inequalities that hold for constant-time slices \cite{Wall:2012uf,Hayden:2011ag,Bao:2015bfa}:

\be
\begin{aligned}
S(AB)+S(BC)&\geq S(B)+S(ABC), \\
S(AB)+S(BC)+S(AC)&\geq S(A)+S(B)+S(C)+S(ABC).
\end{aligned}
\ee
for disjoint regions $A$, $B$, and $C$. The above statements are all proven in detail for maximin surfaces in \Ref{Wall:2012uf}.
\section{A Multi-Wormhole Area Theorem}\label{sec:gravity}

We are now ready to find the gravitational statement dual to entanglement conservation. Let us take as our spacetime $M$ the most general possible setup to consider in the context of the ER=EPR correspondence: an arbitrary, dynamical collection of wormholes and black holes in asymptotically AdS spacetime. We work in $D$ spacetime dimensions. Throughout, we will assume that $M$ obeys the null curvature condition~\eqref{eq:NCC}. The degrees of freedom associated with the Hilbert space ${\cal H} = \otimes_i {\cal H}_i$ can be considered to be localized on the union of the stretched horizons, with each horizon comprising one of the ${\cal H}_i$ factors. We choose our spacetime setup such that the wormholes are past-initialized, by which we mean that for $t\leq 0$ the wormholes are far apart and the spacetime around the wormholes is in vacuum, with negligible back-reaction. Suppose we arbitrarily divide this system into two subsystems by labeling each horizon as ``left'' or ``right''. The left and right Hilbert spaces factorize as ${\cal H}_{\rm L} = \otimes_i {\cal H}_{{\rm L},i}$ and ${\cal H}_{\rm R} = \otimes_i {\cal H}_{{\rm R},i}$, where $H_{{\rm L (R)},i}$ contains the degrees of freedom associated with horizon $i$ in the left (right) set. Now, some of the black holes in the left subset may be entangled with each other and so be described by ER bridges among the left set. A similar statement applies to the right set. Importantly, there may be horizons in the left set entangled with horizons in the right set, describing ER bridges across the left/right boundary. For the sake of tractability, we consider horizons that are only pairwise entangled and that begin in equal-mass pairs in the asymptotically AdS spacetime; this stipulation can be made without loss of generality provided we consider black holes smaller than the AdS length and do not consider changes to the asymptotic structure of the spacetime (see, e.g., \Ref{Balasubramanian:2014hda}). (To treat wormholes with mouths of unequal masses, we could start in an equal-mass configuration and add matter into one of the mouths.) We thus take any two horizons $i$ and $j$ that are entangled to be in the thermofield double state at $t=0$,
\be
\Pi_i \Pi_j \ket{\psi}(t=0) = \ket{\psi_{i,j}}(t=0) = \frac{1}{\sqrt{Z}}\sum_n e^{-\beta E_n/2} \ket{\bar{n}}_i \otimes \ket{n}_j,\label{eq:thermofield}
\ee 
where $\Pi_i$ is a projector onto the degrees of freedom associated with ${\cal H}_i$, $1/\beta$ is the temperature, and $\ket{n}_i$ is the $n^{\rm th}$ eigenstate of the CFT corresponding to the degrees of freedom in ${\cal H}_i$ with eigenvalue $E_n$.

Let us define a time slicing of the spacetime $M$ into spacelike codimension-one surfaces $\Sigma_t$ parameterized by a real number $t$ that smoothly approaches the standard AdS time coordinate in the limit of spacelike infinity, where the metric is asymptotically AdS. The $\Sigma_t$ are chosen to pass through the wormholes without coordinate singularities along the horizon (cf. Kruskal coordinates); see \Fig{fig:PenroseSimple} for an example geometry. For the wormholes spanning the left and right subsets, we write as $L_i$ and $R_i$ the null codimension-one surfaces that form the outermost left and right apparent horizons, respectively, and define $L = \cup_i L_i$ and $R  = \cup_i R_i$. Note that, since new apparent horizons can form outside of the initial apparent horizons, $L_i$ and $R_i$ are each not necessarily connected, but are the piecewise-connected union of the outermost connected components of the apparent horizons. On a given spacelike slice, an apparent horizon is a boundary between regions in which the outgoing orthogonal null congruences are diverging (untrapped) or converging (trapped) \cite{Hawking&Ellis}. Of course, the indexing $i$ may become redundant if horizons merge among the $L_i$ or $R_i$. Let us define the restriction of the outermost apparent horizons to the constant-time slice $\Sigma_t$ as the spacelike codimension-two surfaces $L_{t,i} = L_i \cap \Sigma_t$ and $R_{t,i} = R_i \cap \Sigma_t$ and similarly $L_t = L \cap \Sigma_t$ and $R_t = R \cap \Sigma_t$. Without loss of generality, we will use the initial spatial separation of the wormholes along with diffeomorphism invariance to choose the $\Sigma_t$ and the parameterization of $t$ such that $\Sigma_0$ intersects the codimension-two bifurcation surfaces $B_i \equiv L_{0,i}= R_{0,i}$ at which all the wormholes have zero length. The past-initialization condition then means that the wormholes are far apart in the white hole portion of the spacetime, which corresponds to $t\leq 0$. Throughout, we will assume that $M\cup\partial M$ is globally hyperbolic; equivalently \cite{Geroch:1970uw}, we will assume that the closure of $\Sigma_0$ is a Cauchy surface for $M\cup\partial M$.

\begin{figure}[t]
  \begin{center}
    \includegraphics[width=0.5\textwidth]{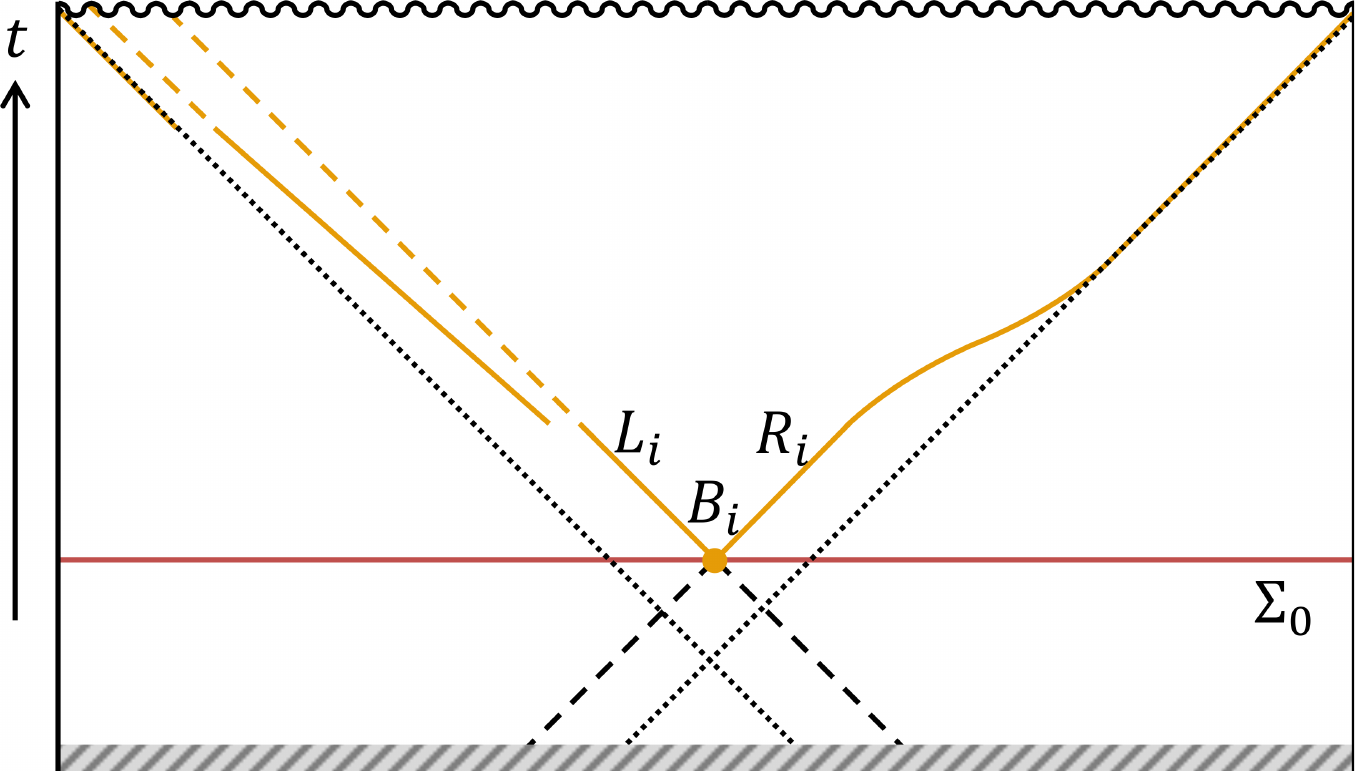}
  \end{center}
  \caption{Penrose diagram, for an example spacetime $M$, of a slice through a particular wormhole $i$ joining a left and right horizon. (Showing the full geometry would require a multi-sheeted Penrose diagram to accommodate the multiple wormholes.) The spacelike codimension-one surface $\Sigma_0$ is shown in burgundy. The initial bifurcation codimension-two surface $B_i$ is illustrated by the orange dot. Apparent horizons are denoted by the orange lines, with the outermost apparent horizons $L_i$ and $R_i$ being the solid lines. For $t\leq 0$, the setup is past-initialized and the metric is given to good approximation by the eternal black hole in AdS, where the past event horizon of the white hole is indicated by the dashed black lines. The dotted black lines denote the future event horizon of $M$. As the spacetime at negative $t$ is known, we do not show the entire Penrose diagram in this region, as indicated by the diagonal gray lines.}
\label{fig:PenroseSimple}
\end{figure}

Now, for each $t>0$, let us define a $D$-dimensional region of spacetime $W_t$ as the union over all achronal surfaces with boundary $L_t \cup R_t$; that is, $W_t$ is the causal diamond associated with $L_t \cup R_t$. A single wormhole has topology $S^{D-2} \otimes \mathbb{R}$ when restricted to $\Sigma_t$. The initial spacetime $W_0$ is special: it is a codimension-two surface that is just the union over all the $B_i$, with topology $(S^{D-2})^{\otimes N}$, where $N$ is the number of wormholes connecting the left and right subsets.

\begin{figure}[H]
  \begin{center}
    \includegraphics[width=0.5\textwidth]{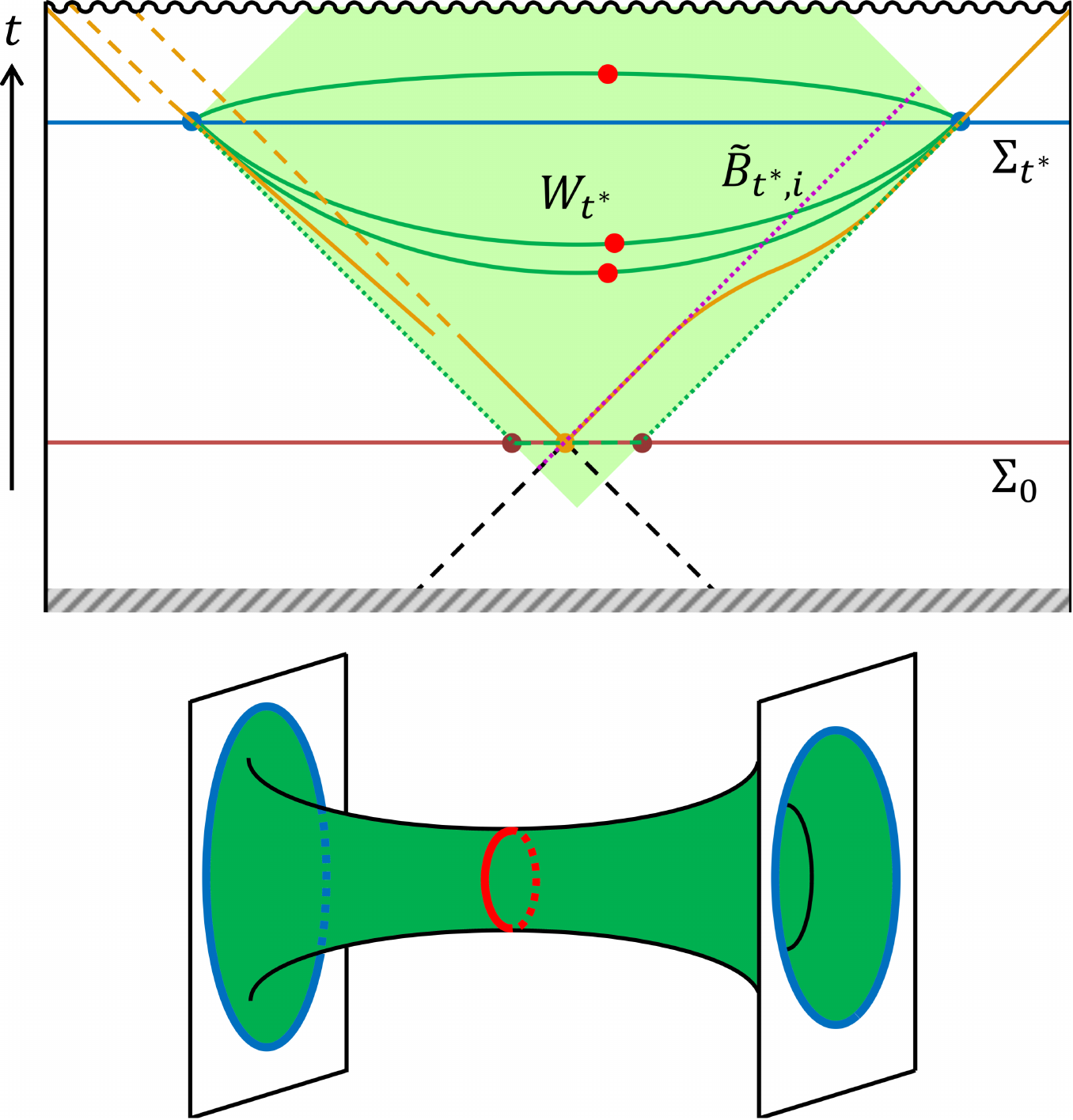}
  \end{center}
  \caption{Penrose diagram (top), for the example geometry of \Fig{fig:PenroseSimple}, of the segment of the region $W_{t^*}$ (green shading), for some $t^*$, that passes through a particular wormhole $i$ joining a left and right horizon. The apparent horizons (orange lines, with solid lines for the outermost apparent horizons $L_i$ and $R_i$), bifurcation surface $B_i$ (orange dot), spacelike codimension-one surface $\Sigma_0$ (burgundy line), and past event horizons for the white hole (dashed black lines) are illustrated as in \Fig{fig:PenroseSimple}. The spacelike codimension-one surface $\Sigma_{t^*}$ is shown as a blue line. The purple dotted line denotes the truncated null surface $\tilde{B}_{t^*,i}$ formed from the rightward outgoing orthogonal null congruence $\tilde{B}_i$ originating on $B_i$, used in Proposition~\ref{Prop1}. The codimension-two boundaries of $W_{t^*}$ along wormhole $i$, $L_{t^*,i}$ and $R_{t^*,i}$, are indicated by the blue dots. The achronal codimension-one surfaces $\Gamma_{t^*}(\alpha)$ foliating $W_{t^*}$ are indicated within wormhole $i$ by the green lines; the codimension-two surfaces $C_{t^*}(\alpha)$ of minimal area for some slices $\Gamma_{t^*}(\alpha)$ are indicated within wormhole $i$ by red dots. The particular surface $\Gamma_{t^*}(0)$, constructed in \Eq{eq:Gammat*0}, is shown (for the portion restricted to wormhole $i$) by the dashed and dotted green lines, corresponding to $\Sigma_0 \cap W_{t^*}$ (the horizontal section) and $M_+ \cap \dot{J}^-[\Sigma_{t^*}\backslash W_{t^*}] = \tilde{L} \cup \tilde{R}$ (the diagonal sections), respectively. The burgundy dots denote the pieces of $\tilde{L}_0$ and $\tilde{R}_0$ in the vicinity of wormhole $i$. The embedding diagram (bottom) shows a particular slice $\Gamma_{t^*}(\alpha)$ through $W_{t^*}$ for some $\alpha$, where, as in the Penrose diagram, the codimension-two boundaries $L_{t^*,i}$ and $R_{t^*,i}$ are shown in blue and the surface $C_{t^*}(\alpha)$ of minimal cross-sectional area, restricted to wormhole $i$, is shown in red.}
\label{fig:Penrose}
\end{figure}

For a given $W_t$, let us define a slicing of $W_t$, parameterized by $\alpha$, with achronal codimension-one surfaces $\Gamma_t(\alpha)$, where the boundary of $\Gamma_t(\alpha)$ is anchored at $L_t \cup R_t$ for all $\alpha$ and where $\alpha$ increases monotonically as we move from the past to the future boundary of $W_t$. Now, we can imagine slicing $\Gamma_t(\alpha)$ into codimension-two surfaces and write as $C_t(\alpha)$ the surface with minimal area [i.e., the minimal cross-sectional area of $\Gamma_t(\alpha)$]; see \Fig{fig:Penrose}. We can now define the maximin surface $C_t$ for $W_t$ as a surface for which the area of $C_t(\alpha)$ attains its maximum under our achronal slicing $\Gamma_t(\alpha)$, maximized over all possible such slicings. That is, $C_t$ is a codimension-two surface with the maximum area, among the set of the surfaces  of minimal cross-sectional area, for all achronal slices through $W_t$. 

The main result that we will prove is that the area of the maximin surface $C_t$ is actually independent of $t$, equaling just the sum of the areas of the initial bifurcation surfaces $B_i$.\footnote{In \Ref{Mukund} it was shown for the special cases of the Schwarzschild-AdS and the single, symmetric, Vaidya-Schwarzschild-AdS geometries that the initial bifurcation surface is the extremal surface in the HRT prescription. Our theorem in this paper generalizes this result to an arbitrary, dynamical, multi-wormhole geometry in asymptotically AdS spacetime that is past-initialized and that obeys the null curvature condition.} In most cases, the maximin surface $C_t$ will actually be the union of the initial bifurcation surfaces $B_i$, independent of $t$. In other words, the maximin area is invariant among all of the different causal diamonds $W_t$. Interpreting the area of the maximin surface as an entropy, this is the gravitational analogue of entanglement conservation. We will first prove a few intermediate results.

\begin{lemma}
\label{Prop1}
The area of the maximin surface $C_t$ is upper bounded by the sum of the areas of the initial bifurcation surfaces $B_i$.
\end{lemma}

\begin{proof}
Consider the rightward outgoing orthogonal null congruence $\tilde{B}_i$, a null codimension-one surface starting on $B_i$ and satisfying the geodesic equation. Choosing some particular $t^*$ arbitrarily, we truncate the null geodesics generating $\tilde{B}_i$ whenever a caustic is reached or when they intersect either the future singularity or the future null boundary of $W_{t^*}$; we further extend the null geodesics into the past until they intersect the past null boundary of $W_{t^*}$. We will hereafter write the truncated null surface as $\tilde{B}_{t^*,i}$. Let $\lambda$ be an affine parameter for $\tilde{B}_{t^*,i}$ that increases toward the future and vanishes on $B_i$; let us write $\tilde{B}_{t^*,i}(\lambda)$ for the spatial codimension-two surface at fixed $\lambda$.
The rotation $\hat{\omega}_{\mu\nu}$ in a space orthogonal to the tangent vector $k^\mu = (\mathrm{d}/{\mathrm{d}\lambda)^\mu}$ satisfies \cite{Sean}
\be
\frac{{\rm D}\hat{\omega}_{\mu\nu}}{{\rm d}\lambda} = -\theta\hat{\omega}_{\mu\nu}, 
\ee
where $\theta = \nabla_\mu k^\mu$ is the expansion. Since $\theta$ vanishes on $B_i$, $\hat{\omega}_{\mu\nu}$ vanishes identically on $\tilde{B}_{t^*,i}$. The Raychaudhuri equation is therefore
\be
\frac{{\rm d}\theta}{{\rm d}\lambda} = -\frac{1}{D-2}\theta^2 - \hat{\sigma}_{\mu\nu}\hat{\sigma}^{\mu\nu} - R_{\mu\nu}k^\mu k^\nu,\label{eq:Raychaudhuri}
\ee
where $\hat{\sigma}_{\mu\nu}$ is the shear and $R_{\mu\nu}$ is the Ricci tensor. We note that if the null curvature condition \eqref{eq:NCC} is satisfied, then $\theta$ is nonincreasing, as $\hat{\sigma}_{\mu\nu}\hat{\sigma}^{\mu\nu}$ is always nonnegative. Since the apparent horizon consists of marginally outer trapped surfaces (i.e., surfaces for which the outgoing orthogonal null geodesics have $\theta = 0$), it must be either null or spacelike, so any orthogonal null congruence starting on the apparent horizon remains either on or inside the apparent horizon in the future \cite{Hawking&Ellis}. In particular, $\tilde{B}_{t^*,i}\subset W_{t^*}$. 

Now, we can also write $\theta$ as ${\rm d} \log \delta A/{\rm d} \lambda$, where $\delta A$ is an infinitesimal cross-sectional area element of $\tilde{B}_{t^*,i}(\lambda)$. That is, $\mathrm{area}[\tilde{B}_{t^*,i}(\lambda)]$ has negative second derivative in $\lambda$. Since $\theta$ vanishes on the bifurcation surface $B_i=\tilde{B}_{t^*,i}(0)$, we have that $\mathrm{area}[\tilde{B}_{t^*,i}(\lambda)]$ is monotonically nonincreasing in $\lambda$. Moreover, since for all $\lambda<0$ there exists $t<0$ such that $\tilde{B}_{t^*,i}(\lambda)\subset\Sigma_{t}$, the past-initialization condition means that $\mathrm{area}[\tilde{B}_{t^*,i}(\lambda)] = \mathrm{area}[B_i]$ for all $\lambda<0$. Hence, for all $\lambda$ we have
\be 
\mathrm{area}[\tilde{B}_{t^*,i}(\lambda)] \leq \mathrm{area}[B_i].\label{eq:decreasingarea}
\ee 

By the past-initialization condition, there are no caustics to the past of $B_i$. Further, by definition, the wormhole does not pinch off until the singularity is reached, so some subset of the generators of $\tilde{B}_i$ must extend all the way through $W_{t^*}$ without encountering caustics. Writing $\Gamma_{t^*}(\alpha)$ as a foliation of $W_{t^*}$ by achronal slices, we thus have that $\tilde{B}_{t^*,i}(\lambda) \cap \Gamma_{t^*}(\alpha)$ is never an empty set for all $\alpha$, i.e., for all $\lambda$ there exists $\alpha$ such that $\tilde{B}_{t^*,i}(\lambda)\subset\Gamma_{t^*}(\alpha)$. Moreover, we can reparameterize and identify the affine parameters for each $i$ of the $\tilde{B}_{t^*,i}$ such that for each $\lambda$ there exists $\alpha$ for which $\cup_i\tilde{B}_{t^*,i}(\lambda)\subset\Gamma_{t^*}(\alpha)$; for such $\alpha$, $\cup_i \tilde{B}_{t^*,i}$ is a complete cross-section of $\Gamma_{t^*}(\alpha)$, possibly with redundancy due to merging horizons. We choose our slicing $\Gamma_{t^*}(\alpha)$ such that there exists some $\alpha^*$ for which $\Gamma_{t^*}(\alpha^*)$ contains the maximin surface $C_{t^*}$ for $W_{t^*}$, so
\be
C_{t^*} = C_{t^*}(\alpha^*)\text{ such that } \mathrm{area}[C_{t^*}(\alpha^*)] = \max_\alpha \mathrm{area}[C_{t^*}(\alpha)],\label{eq:Ct*}
\ee
where $C_{t^*}(\alpha)$ is the codimension-two cross-section of $\Gamma_{t^*}(\alpha)$ with minimal area.

Since $\tilde{B}_{t^*,i}$ is only completely truncated at future and past boundaries of $W_{t^*}$, it follows that for every $\alpha$ there must exist $\lambda$ such that $\Gamma_{t^*}(\alpha)\supset \tilde{B}_{t^*,i}(\lambda)$. By the definition of $C_{t^*}(\alpha)$, we have (for such $\lambda$) that
\be
\mathrm{area}[C_{t^*}(\alpha)] \leq \sum_i \mathrm{area}[\tilde{B}_{t^*,i}(\lambda)].\label{eq:minsurface}
\ee
Putting together \Eqs{eq:decreasingarea}{eq:minsurface}, taking the maximum over $\lambda$ and $\alpha$ on both sides, applying \Eq{eq:Ct*}, and using the fact that $t^*$ was chosen arbitrarily, we have a $t$-independent upper bound on the area of the maximin surface $C_t$:
\be 
\mathrm{area}[C_t]\leq \sum_i \mathrm{area}[B_i].\label{eq:proofLEQ}
\ee
\end{proof}

Let us now construct a lower bound on the area of the maximin surface $C_t$. We can do this by examining an achronal codimension-one surface through $W_t$ and computing its minimal cross-sectional area; judiciously choosing the achronal surface optimizes the bound. In particular, for some arbitrary $t^*$, consider $\Gamma_{t^*}(0)$ passing through $\cup_i B_i$, where we choose the slicing such that
\be
\Gamma_{t^*}(0) = (\Sigma_0 \cap W_{t^*}) \cup \left(M_+ \cap \dot{J}^-[\Sigma_{t^*}\backslash W_{t^*}]\right),\label{eq:Gammat*0}
\ee
where $M_+$ is the restriction of $M$ to $t\geq 0$, $J^-[A]$ denotes the causal past of a set $A$, and the dot denotes its boundary. That is, $\Gamma_{t^*}(0)$ consists of the codimension-one null surfaces forming the $t\geq 0$ portion of the boundary of $W_{t^*}$ towards the past, plus a codimension-one segment of $\Sigma_0$ containing $\cup_i B_i$; see \Fig{fig:Penrose}. Let us label the left and right boundaries of $\Sigma_0 \cup W_{t^*}$ (equivalently, the left and right portions of the intersection of $\Sigma_0$ and $\dot{J}^-[\Sigma_{t^*}\backslash W_{t^*}]$) as $\tilde{L}_0$ and $\tilde{R}_0$, respectively.

We will show in two steps that the minimal cross-sectional area of $\Gamma_{t^*}(0)$ is just $\sum_i \text{area}[B_i]$. We will first consider the cross-sectional area of slices of $\Sigma_0\cap W_{t^*}$ and then examine the changes in cross-sectional area along slices of $M_+ \cap \dot{J}^-[\Sigma_{t^*}\backslash W_{t^*}]$.

\begin{lemma}
\label{Prop2}
The minimal cross-sectional area of $\Sigma_0\cap W_{t^*}$ is $\sum_i \mathrm{area}[B_i]$.
\end{lemma}

\begin{proof}
By the requirement that the wormholes be past-initialized, the metric on $\Sigma_0$ is, up to negligible back-reaction, just a number of copies of the metric on the $t=0$ slice of the single maximally-extended AdS-Schwarzschild black hole; for this metric the $t_{\rm KS} = 0$ and $t_{\rm S} = 0$ slices are the same, where $t_{\rm KS}$ is the Kruskal-Szekeres time coordinate and $t_{\rm S}$ is the Schwarzschild time coordinate \cite{Bao:2015nca}. Taking the $t$-slicing to correspond to the Kruskal-Szekeres coordinates in the vicinity of each wormhole, therefore, the metric on $\Sigma_0 \cap W_{t^*}$ is 
\be
{\rm d}s^2_{\Sigma_0 \cap W_t} = \frac{4|f(r)|e^{-f'(r_{\rm H})r^*}}{[f'(r_{\rm H})]^2}{\rm d}X^2 + r^2 {\rm d}\Omega_{D-2}^2 =  \frac{{\rm d}r^2}{f(r)} + r^2 {\rm d}\Omega_{D-2}^2,
\ee
where on $\Sigma_0$, the Kruskal $X$ coordinate describing distance away from the wormhole mouth at $B_i$ is $X = \pm e^{f'(r_{\rm H}) r^*/2}$, with the sign demarcating the left and right side of $B_i$ and the tortoise coordinate being $r^* = \int {\rm d} r/f(r)$. The function $f(r)$ is
\be
f(r) = 1 - \frac{16\pi G_D M}{(D-2)\Omega_{D-2}r^{D-3}} + \frac{r^2}{\ell^2}, 
\ee
where $\Omega_{D-2}$ is the area of the unit $(D-2)$-sphere, $G_D$ is Newton's constant in $D$ dimensions, $M$ is the initial mass of each wormhole mouth, $\ell$ is the AdS length, and $r_{\rm H}$ is the initial horizon radius, defined such that $f(r_{\rm H}) = 0$.  For $r>r_{\rm H}$, $f(r)$ is strictly positive, so $r^*$ and $X$ are monotonic in $r$. As we move from $B_i$ at $X = 0$ towards $\tilde{L}_0$ or $\tilde{R}_0$ at $X_{\rm L}$ and $X_{\rm R}$, the area of the cross-section of $\Sigma_0\cap W_{t^*}$ for the surface parameterized by $X(\phi)$ [or equivalently $r(\phi)$], for $(D-2)$ angular variables $\phi$, attains its minimum at $B_i$, where $r(\phi)$ is identically $r_{\rm H}$, its minimum on $\Sigma_0 \cap W_{t^*}$.
\end{proof}

We now turn to the behavior of the cross-sectional area of $M_+ \cap \dot{J}^-[\Sigma_{t^*}\backslash W_{t^*}]$.

\begin{lemma}
\label{Prop3}
The cross-sectional area of $M_+\cap\dot{J}^-[\Sigma_{t^*}\backslash W_{t^*}]$ is nondecreasing towards the future.
\end{lemma}

\begin{proof}
Let us label the left and right halves of $M_+\cap\dot{J}^-[\Sigma_{t^*}\backslash W_{t^*}]$ as $\tilde{L}$ and $\tilde{R}$, so the boundary of $\tilde{L}$ is just $\tilde{L}_0 \cup L_{t^*}$ and similarly for $\tilde{R}$. We note that both $\tilde{L}$ and $\tilde{R}$ are generated by outgoing null geodesics. Suppose that some segment of $M_+\cap\dot{J}^-[\Sigma_{t^*}\backslash W_{t^*}]$ has area decreasing towards the future. We can without loss of generality restrict to the left null surface, which we then assume has decreasing area along some segment.

We first observe that since the apparent horizons are null or spacelike and since $\tilde{L}$ is part of the null boundary of the past of a slice through the outermost apparent horizon, all outer trapped surfaces must lie strictly inside $\tilde{L}\cap \Sigma_t$ for all spacelike slices $\Sigma_t$ for $t\in[0,t^*]$.

Let us define an affine parameter $\tilde{\lambda}$ for $\tilde{L}$, for which $\tilde{\lambda} = 0$ on $\tilde{L}_0$ and $\tilde{\lambda} = 1$ on $L_{t^*}$, and consider the expansion $\tilde{\theta} = \nabla_\mu \tilde{k}^\mu$, where $\tilde{k}^\mu = (\mathrm{d}/\mathrm{d}\tilde{\lambda})^\mu$. In order for the area to be strictly decreasing, there must be some open set $U$ for which $\tilde{\theta}(\tilde{\lambda})<0$ for $\tilde{\lambda} \in U$. By continuity of the spacetime, there must exist $\tilde{t}$, where we can choose the affine parameterization such that $\Sigma_{\tilde{t}}\supset \tilde{L}(\tilde{\lambda})$ for some $\tilde{\lambda}\in U$, such that $\Sigma_{\tilde{t}}$ contains a region $V\supset \tilde{L}(\tilde{\lambda})$ for which $\tilde{\theta}\leq 0$ for all outgoing orthogonal null congruences originating from $V$. Then $V$ is an outer trapped surface not strictly inside $\tilde{L}\cap\Sigma_{\tilde{t}}$. This contradiction completes the proof.
\end{proof}

Thus, we have constructed a lower bound for the area of $C_t$.

\begin{lemma}
\label{Prop4}
The area of $C_t$ is lower bounded by the sum of the areas of the initial bifurcation surfaces $B_i$.
\end{lemma}

\begin{proof}
To prove a lower bound on the maximin area, $\text{area}[C_{t^*}]$, it suffices to exhibit an achronal surface through $W_{t^*}$ for which the minimal cross-sectional area is equal to the desired lower bound. Such a surface is given by $\Gamma_{t^*}(0)$ in \Eq{eq:Gammat*0}: by Proposition~\ref{Prop2}, $\sum_i \text{area}[B_i]$ is the minimal cross-sectional area of $\Sigma_0 \cap W_{t^*}$ and, in particular, $\sum_i \text{area}[B_i] \leq \text{area}[\tilde{L}_0] + \text{area}[\tilde{R}_0]$. By Proposition~\ref{Prop3}, the minimal cross-sectional area of $M_+\cap\dot{J}^-[\Sigma_{t^*}\backslash W_{t^*}]$ is $\text{area}[\tilde{L}_0] + \text{area}[\tilde{R}_0]$. Thus, $\Gamma_{t^*}(0)$ is an achronal slice through $W_{t^*}$ with minimal cross-sectional area equal to $\sum_i \text{area}[B_i]$.
\end{proof}

Finally, as an immediate corollary, we have the gravity dual of entanglement conservation.

\begin{theorem} 
For the family of spacetime regions $W_t$ defined as the causal diamonds anchored on the piecewise-connected outermost apparent horizons $L_t$ and $R_t$ for an arbitrary set of dynamical, past-initialized wormholes and black holes satisfying the null curvature condition, the corresponding maximin surface $C_t$ dividing the left and right collections of wormholes has an area independent of $t$, equaling the sum of the areas of the initial bifurcation surfaces for the wormholes linking the left and right sets of horizons.
\end{theorem}

\begin{proof}
By Proposition~\ref{Prop1}, $\text{area}[C_t]\leq\sum_i\text{area}[B_i]$, while by Proposition~\ref{Prop4}, $\text{area}[C_t]\geq\sum_i\text{area}[B_i]$. Hence,
\be
\text{area}[C_t] = \sum_i \text{area}[B_i]. 
\ee
\end{proof}
Thus, the maximin surface dividing one collection of wormhole mouths from another has an area that is conserved under arbitrary spacetime evolution and horizon mergers as well as arbitrary addition of matter satisfying the null energy condition. Viewing the maximin surface area as the entanglement entropy associated with the left and right sets of horizons in accordance with the HRT prescription, we have proven a statement in general relativity that is a precise analogue of the statement in \Sec{sec:entanglement} of conservation of entanglement under evolution of a state with a tensor product unitary operator.

\section{Conclusions}\label{sec:Conclusions}

The proposed ER=EPR correspondence is surprising insofar as it identifies a generic feature (entanglement) of any quantum mechanical theory with a specific geometric and topological structure (wormholes) in a specific theory with both gravity and spacetime (quantum gravity). Until an understanding is reached of the geometrical nature of the ``quantum wormholes'' that should be dual to, e.g., individual entangled qubits, it will be difficult to directly establish the validity of the ER=EPR correspondence as a general statement about quantum gravity. In a special limiting case of quantum gravity---namely, the classical limit, which gives general relativity---this task is more tractable. In this paper, we have provided a general and explicit elucidation of the ER=EPR correspondence in this limit. For a spacetime geometry with an arbitrary set of wormholes and black holes, we have constructed the maximin area of the multi-wormhole throat separating a subset of the wormholes from the rest of the geometry, the analogue of the entanglement entropy of a reduced density matrix constructed from a subset of the degrees of freedom of a quantum mechanical state. We then proved that the maximin area is unchanged under all operations that preserve the relation between the subset and the rest of the geometry, the equivalent of quantum mechanical operations that leave the entanglement entropy invariant. We have therefore completely characterized the ER=EPR relation in the general relativistic limit: the entanglement entropy and area (in the sense defined above) of wormholes obey precisely the same rules. 

In addition to providing an examination of the ER=EPR duality, our result constitutes a new area theorem within general relativity. The maximin area of the wormhole throat is invariant under dynamical spacetime evolution and the addition of classical matter satisfying the null energy condition. The dynamics of wormhole evolution were already constrained topologically (see \Ref{Bao:2015nqa} and references therein), but this result goes further by constraining them geometrically. Note that throughout this paper we have worked in asymptotically AdS spacetimes in order to relate our results to a boundary theory using the language of the AdS/CFT correspondence, but our area theorem is independent of this asymptotic choice provided that all of the black holes are smaller than the asymptotic curvature scale.

In the classical limit, we have characterized and checked the consistency of the ER=EPR correspondence in generality. However, extending these insights to a well-defined notion of quantum spacetime geometry and topology remains a formidable task. Understanding the nature of the ER=EPR duality for fully quantum mechanical systems suggests a route toward addressing the broader question of the relationship between entanglement and geometry.

\begin{center} 
 {\bf Acknowledgments}
 \end{center}
 \noindent 

We thank Sean Carroll and Mukund Rangamani for useful discussions and comments. This research was supported in part by DOE grant DE-SC0011632 and by the Gordon and Betty Moore Foundation through Grant 776 to the Caltech Moore Center for Theoretical Cosmology and Physics. N.B. is supported by the DuBridge postdoctoral fellowship at the Walter Burke Institute for Theoretical Physics. G.N.R. is supported by a Hertz Graduate Fellowship and a NSF Graduate Research Fellowship under Grant No.~DGE-1144469.

\bibliographystyle{utphys}

\bibliography{EPRConsBib}

\end{document}